\documentclass[runningheads]{llncs}
\usepackage[T1]{fontenc}
\usepackage{graphicx}
\usepackage{hyperref}
\usepackage{color}

\usepackage{amsmath}
\usepackage{amssymb}
\usepackage{xspace}
\usepackage{algorithm}
\usepackage[noend]{algpseudocode}
\usepackage{tikz}
\usetikzlibrary{positioning,hobby,calc}

\newcommand{\bc}[1]{\ensuremath{\overline{#1}^{_{bip}}}}
\newcommand{\D}[1]{MBS_{#1}}
\newcommand{\DC}[1]{MBS^C_{#1}}
\title{Maximum Biclique for Star$_{1,2,3}$-free and Bounded Bimodularwidth Twin-free Bipartite Graphs\thanks{Supported by French ANR 20-CE23-0002 Coregraphie}}
\titlerunning{Maximum Biclique for $Star_{123}$-free and Bounded bmw Bipartite Graphs}
\author{Fabien de Montgolfier
  \and Renaud Torfs}
\authorrunning{F. de Montgolfier and R. Torfs}
\institute{IRIF, Université Paris Cité, France.
  \email{fm@irif.fr} \email{torfs@irif.fr}  \\
}
\begin{document}

\maketitle              
\begin{abstract}
  There are three usual definitions of a maximum bipartite clique (biclique) in a bipartite graph : either maximizing the number of vertices, or of edges, or finding a maximum balanced biclique. The first problem can be solved in polynomial time, the last ones are NP-complete. Here we show how these three problems may be efficiently solved for two classes of bipartite graphs:  $Star_{123}$-free twin-free graphs, and bounded bimodularwidth twin-free graphs, a class that may be defined using bimodular decomposition. Our computation requires $O(n^2)$ time and requires a decomposition is provided, which takes respectively $O(n+m)$ and $O(mn^3)$ time.

\end{abstract}
%
%

\section{Introduction}
This paper addresses the problem of computing a \emph{maximum} bipartite clique (biclique) in a bipartite graph. While the problem of computing a maximum clique in a graph is well defined, the one of a maximum biclique in a bipartite graph, however, has (at least) three nonequivalent definitions :
\begin{itemize}
\item either maximizing the number of vertices (Vertex-Maximum Biclique),
\item or maximizing the number of edges (Edge-Maximum Biclique),
\item  or finding a biclique of maximum cardinality with the same number of white and of black vertices (Maximum Balanced Biclique)
\end{itemize}
A variation of the Vertex-Maximum Biclique found sometimes is that it must contains an edge (Non-trivial Vertex-Maximum Biclique).
Among the many applications of these problems,  we may cite gene expression analysis~\cite{yetanotherbio} and other various problems from Biology (see~\cite{pavlopoulos2018bipartite} for a survey),
anomaly detection in crowdsourcing~\cite{collusion} or in social networking~\cite{Beutel2013CopyCatchSG},
modeling complex networks of various kinds~\cite{GUILLAUME2006795}, and clustering~\cite{BicliqueComm}.

Garey and Johnson address the first and third problem (\cite{GareyJ79}, problem GT24) proving that Maximum Balanced Biclique is NP-complete (Maximum Clique reduces to it) and notice that the vertex-maximum biclique is polynomial.
Indeed, by König's theorem, in a bipartite graph, a maximum matching has the same size than a minimum vertex cover, and the vertices not in a minimum vertex cover form an maximum independent set. Since the bipartite complement transforms independent sets into bicliques, one just has to run a maximum matching algorithm on the bipartite complement to solve the vertex-maximum biclique problem. Since the bipartite complement of a sparse bipartite graph may have $\Omega(n^2)$ edges, Hopcroft-Karp algorithm runs in $O(n^{2.5})$ and is, as far as we know, the fastest algorithm.
Finally, the proof that the second problem is also NP-complete was published in 2003 only by Peeters~\cite{PEETERS2003651}.

Approximation is hard: Manurangsi proved that, assuming Small Set Expansion Hypothesis  and that NP $\ne$ BPP, for every $\epsilon >0$, no polynomial time algorithm gives $n^{1 - \epsilon}$-approximation for the edge-maximum biclique~\cite{manurangsi2017inapproximability}.
Feige proved that that there is a constant
$\delta > 0$ such that the maximum balanced biclique problem cannot be approximated within a ratio
below $n^\delta$, under the random 3-SAT hardness hypothesis~\cite{Feige02}. Khot proved, assuming that NP $\not\subseteq\cup_{\epsilon > 0}$
BPTIME($2^{n^\epsilon}$), that maximum balanced biclique has no polynomial time approximation scheme (PTAS)~\cite{Khot06}. 
Dawande~\emph{et al.}  survey the weighted and the multipartite extensions of maximum biclique, showing most of them are NP-complete~\cite{DAWANDE2001388}.

In the present paper, we use a variation of the modular decomposition, suitable for bipartite graphs, the so-called \emph{bimodular decomposition}~\cite{WG04}, to solve efficiently these three problems (and in fact more, as we can find any given size biclique if it exists) on two classes of bipartite graphs that behave well with respect to bimodular decomposition.

The paper is structured as follow: first we define bimodular decomposition, then we introduce the dynamic programming tool we use, the Maximum Bisize Set (MBS). In the fourth section we study how the MBS behave with respect to the recursive and base cases of bimodular decomposition. Finally we present two $O(n^2)$-time algorithms solving the maximum biclique problems for the two graph classes we consider: twin-free $Star_{1,2,3}$-free graphs~\cite{lozin2002bipartite}, and twin-free Bounded Bimodularwidth bipartite graphs, a class we introduce here. These algorithms need a decomposition tree to be provided, that can be computed in $O(n+m)$ time in the first case~\cite{quaddoura2002linear}
and in $O(mn^3)$ for  all bipartite graphs \cite{WG04}.

\section{Bimodular decomposition}
In this section, we present our notations, and then four decomposition operations, and use them to present the \emph{bimodular decomposition}.
Thorough this paper, $G=(B\uplus W,E)$ is a bipartite graph where the partition between the white $W$ and black $B$ vertex-set is given. 
We denote $V=B\cup W$, $n=|V|$ and $m=|E|$.
For $X\subset V$ we denote $B_X=B \cap X$ and $W_X=W\cap X$. 
The \emph{bipartite complement} of $G=(B\uplus W,E)$ is $\bc{G}=(B\uplus W,(B\times W)-E)$.
For two disjoint subsets $X,Y\subset V$, we say
\begin{itemize}
\item $X$ is \emph{nonadjacent} to $Y$ if there is no edge between a vertex from $X$ and a vertex from $Y$.
\item $X$ is \emph{left adjacent} to $Y$ when every black vertex from $X$ is adjacent to every white vertex from $Y$ and no white vertex from $X$ has a black neighbor in $Y$.
\item $X$ is \emph{fully adjacent} to $Y$ when every black (resp. white)  vertex from $X$ is adjacent to every white (resp. black) vertex from $Y$.
\end{itemize}
  A vertex $x$ is \emph{isolated} if $\{x\}$ is nonadjacent to $V-x$ and \emph{universal} if  $\{x\}$ is fully adjacent to $V-x$.

Let us suppose that $V$ is partitioned into nonempty and disjoint  $V_1...V_k$, $k\ge 2$. Let us denote $W_i=W\cap V_i$, $B_i=B\cap V_i$, and $G_i=G[V_i]$.
\begin{itemize}
\item If $\forall i\ne j$ $V_i$ is \emph{nonadjacent} with $V_j$, and $k$ is maximum among partitions having this property, then we say that $G$ admits a \emph{Parallel decomposition} into $G_1...G_k$. Notice $G_i$  is a connected component of $G$.
\item If $\forall i\ne j$ $V_i$ is \emph{fully adjacent} with $V_j$, and $k$ is maximum among partitions having this property, then we say that $G$ admits a \emph{Series decomposition} into $G_1...G_k$. Notice $G_i$ is a connected component of $\bc{G}$.
\item If  $\forall i<j$ $V_i$ is \emph{left-adjacent} with $V_j$, and $k$ is maximum among partitions having this property, then we say that $G$ admits a \emph{K+S decomposition} into $G_1...G_k$. Notice that for each $i$,
  $B_1\cup ...\cup B_i \cup W_{i+1}\cup...\cup W_k$ is a biclique, while
  $W_1\cup ...\cup W_i \cup B_{i+1}\cup...\cup B_k$ is a stable set, hence the ``K+S'' decomposition name.
\end{itemize}
  In the case $G$ admits both a K+S and a Parallel (resp. Series) decomposition, then $G$ has an isolated (resp. universal) vertex. For shake of unicity, in this case we define that $G$ admits only a K+S decomposition.

  A bipartite graph is twin-free when two vertices can not have the same neighborhood. A $Star_{123}$, also called \emph{Skew Star}, is seven-vertex graph consisting of  a path of six vertices plus a pending vertex adjacent to the third vertex of that path.
  
\begin{theorem}[Lozin~\cite{lozin2002bipartite}]\label{thlozin}
  Let $G$ be a twin-free bipartite graph without induced $Star_{1,2,3}$. Then
  \begin{itemize}
  \item either $G$ admits a Parallel decomposition,
  \item  or $G$ admits a Series decomposition,
  \item or $G$ admits a K+S decomposition,
  \item or $G$ is $K_{1,3}$-free
  \item or $\bc{G}$ is $K_{1,3}$-free
  \end{itemize}
  \end{theorem}

Notice the base cases are that $G$ is either a  path, or a cycle, or the bipartite complement of a path or a cycle. This theorem was later extended to exclude no induced graph. The key idea to do so is to use  bimodules. A \emph{bimodule} is a set $M$ of vertices such that every vertex of $V-M$ is either nonadjacent or fully adjacent to $M$. A bimodule is trivial when it has at most one black vertex and at most one white vertex. In \cite{WG04} is proven that, when a graph has no Parallel nor Series nor K+S decomposition, then two maximal nontrivial bimodules $M$ and $M'$ either are disjoint or overlap on only one so-called \emph{augmenting vertex} $v$, that is either nonadjacent or fully adjacent to $M$ and $M'$. Removing the (at most 2) augmenting vertices from each maximal nontrivial bimodule yields the maximal \emph{canonical} bimodules, that do not overlap. Therefore the vertex-set of a nontrivial twin-free bipartite graph with no Parallel nor Series nor K+S decomposition can be uniquely partitioned into  maximal canonical bimodules, plus the other vertices (each vertex not in a maximal canonical bimodule forming a singleton class), yielding the Prime  decomposition case. 

\begin{theorem}[Fouquet \emph{et. al}~\cite{WG04}]\label{thfm}
   Let $G$ be a twin-free bipartite graph. Then
  \begin{itemize}
  \item either $G$ admits a Parallel decomposition,
  \item  or $G$ admits a Series decomposition,
  \item or $G$ admits a K+S decomposition,
  \item or $G$ admits a Prime decomposition (into maximal canonical bimodules and singleton vertices),
  \item or $G$ has only one vertex.
  \end{itemize}
  \end{theorem}

Both Theorems~\ref{thlozin} and \ref{thfm} allow to define a decomposition tree, whose root is labeled by the decomposition case that applies (Series, Parallel, K+S or Prime), and internal nodes correspond to the decomposition of the graphs induced by each component (or maximal canonical bimodule) of the corresponding case. The leaves are the base cases (single vertices or $K_{1,3}$-free graphs or their bipartite complement). For Theorem~\ref{thfm} this tree is called \emph{canonical decomposition tree} and, for each node,  the leaves of the subtree rooted at that node is a bimodule. These bimodules form a  family of non-overlapping bimodules called the canonical bimodules. If a graph has no nontrivial bimodule, it is called \emph{bimodule-prime} and has a Prime decomposition into a trivial partition of $n$ singletons.

\begin{definition}[bimodularwidth]
	The \emph{bimodularwidth} of a twin-free bipartite graph  $G$ is the largest number of children (counting leaves) of a Prime node of the canonical decomposition, or is $2$ if that tree has no Prime node.
\end{definition}

\begin{figure}[h!]\vspace{-.5cm}
	\begin{center}
		\includegraphics[width=.8\textwidth]{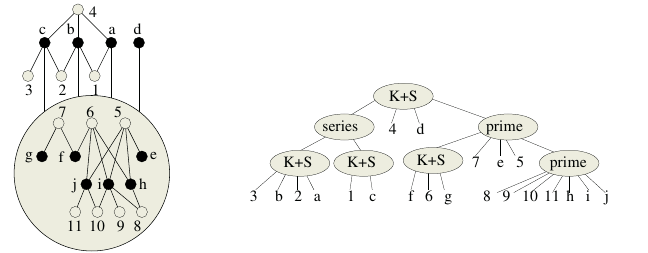}
	\end{center}\vspace{-.5cm}
	\caption{A bipartite graph and its bimodular decomposition tree. Bimodularwidth is 7}\vspace{-1cm}
\end{figure}

\section{Maximum Biclique Size set}
\begin{definition}[Domination]
Given two integer couples $(x,y)$ and $(x',y')$, $(x,y)$ \emph{dominates} $(x',y')$ when $x'\le x$ and $y'\le y$. Domination is \emph{strict} when $x\ne x'$ or $y\ne y'$. For $X\subset   \mathbb{N}^2$, let $Dom(X)$ be the couples of $X$ not strictly dominated by another couple of $X$.
\end{definition}

\begin{definition}[Bisize and maxbisize]
  $(b,w)\in \mathbb{N}^2$ is a \emph{bisize of $G$}  if $G$ contains a biclique of $b$ black and $w$ white vertices.
    A bisize $(b,w)$ of $G$ is a \emph{maxbisize of $G$} if it not strictly dominated by another bisize of $G$.
    A bisize or a maxbisize $(b,w)$ of $G$ is \emph{trivial} when $b=0$ or $w=0$.
  \end{definition}
\begin{remark}
   For each maxbisize there is a maximal biclique of that size, but the converse is not true.
  If a bisize $(b,w)$ dominates $(b',w')$,  $(b',w')$ also is a bisize, by inclusion of bicliques.
 \end{remark}
\begin{definition}[maxbisize set and operators]
  \begin{itemize}
  \item The \emph{maxbisize set} of $G$, denoted $\D G$, is the set of all maxbisizes of $G$.
  \item For $M\subset V$,    $\D M$ denotes the maxbisize set of $G[M]$.
\item 	Let $(b,w)+(b',w')$ be $(b+b',w+w')$. For two maxbisize sets $X$ and $Y$ let $X \oplus Y$ be $ \{ x+y\ | x \in X\ y\in Y\}$. If one side is empty, define  $\emptyset\oplus X=X\oplus \emptyset=X$.
\item Let $\rightarrow_w^zX$ be $\{(x,y+z) | (x,y) \in X\}$ and let $\rightarrow_b^z$ be $\{(x+z,y) | (x,y) \in X\}$
  \end{itemize}
\end{definition}

We may reduce our three biclique problems to computing the maxbisize set:
\begin{theorem}\label{th3cl}
  
  Let $G$ be a bipartite graph
  \begin{itemize}
  \item The vertex-maximum biclique has $\max_{(x,y) \in \D G} x+y$ vertices
  \item The edge-maximum biclique has    $\max_{(x,y) \in \D G} x*y$ edges
  \item The maximum balanced biclique has     $\max_{(x,y) \in \D G} min(x,y)$ vertices of each color
  \end{itemize}
\end{theorem}

\begin{proof}
  Both vertex-maximum and edge-maximum bicliques are maximal biclique so their size are in $\D G$. Counting vertices of a $(b,w)$ biclique is just adding $b$ and $w$, and counting edges is multiplying them, thus the two first assertions. For a maximum balanced biclique of size $(m,m)$ notice that, if that biclique is not maximal, then it is either included in a $(m,w>m)$ or in a $(b>m,m)$ maximal biclique. Taking, for each maxbisize $(b,w)$, $min(b,w)$ thus yields the largest balanced biclique in contains, and thus the largest of all of them is the maximum balanced biclique.
\end{proof}

\begin{proposition}
	\label{mod_flag}
	Let $(b,w)$ be a maxbisize of $G$, $C$ a biclique of size $(b,w)$, and $M$ be a vertex subset  fully adjacent to $C \setminus M$. $(|B_M \cap C|,|W_M \cap C|)$ is in $\D{M}$.
\end{proposition}

\begin{proof}
	$C \cap M$ is a biclique so $(|B_M \cap C|,|W_M \cap C|)$ is a bisize from $G$ and from $G[M]$. Let us suppose it is not a maxbisize of $G[M]$. Then $G[M]$ contains a biclique $C'$ of size at least  $(|B_M \cap C|+1,|W_M \cap C|)$ or $(|B_M \cap C|,|W_M \cap C|+1)$. Since $C$ is fully adjacent to $C\setminus M$, $C' \cup (C \setminus M)$ is a biclique of $G$ of size at least $(b+1,w)$ or $(b,w+1)$ and therefore $(b,w)$ is not a maxbisize of $G$, a contradiction.
\end{proof}

\begin{proposition}
	\label{mod_flag_incl}
Let $(b,w)$ be a maxbisize of $G$, $C$ a biclique of size $(b,w)$, and $M$ be a vertex subset such that $C \subseteq M$. $(b,w)$ is in $\D{M}$.
\end{proposition}

\begin{proposition}\label{remincl}
  If, for some set $M$, $X_M$ is a set of bisizes of $G[M]$ containing all maxbisizes, then $Dom(X_M)=\D M$.
\end{proposition}

\begin{lemma}\label{lemDomn2}
  Let $(k_1,k_2)\in  \mathbb{N}^2$ and $X\subset   \mathbb{N}^2$ such that each element of $X$ is dominated by   $(k_1,k_2)$. $Dom(X)$ can be computed in  $O(|X| + \min(k_1,k_2))$ time.
\end{lemma}

\begin{proof}
Let us suppose wlog. that $k_1 \le k_2$. The algorithm simply creates a zero-filled array $a$ of length  $k_1$. Then for each  $(b,w) \in X$,  $a[b] := \max(a[b],w)$. Finally the pairs ($b, a[b])$ for which  $\forall b' > b$,  $ a[b'] < a[b]$ are output. A simple right-to-left swap of $a$ is enough to check that this condition holds for all entries.

Let us show a pair $(b,w)\in X$ is in $Dom(X)$ iff  $a[b] = w$ and $\forall b' > b$ we have $ a[b'] < w$. If $(b,w)\in Dom(X)$ the max assignments assert that  $a[b]\ge w$. Since $X$ does not contains $(b,w')$ with $w'>w$ we have $a[b]=w$.  If there would exist $b'>b$ such that $a[b'] \ge w$ then  $(b',a[b']) \in X$ dominates $(b,w)$, a contradiction. So we have the direct sense.
For the converse, let us suppose  $a[b]=w$ and $\forall b' > b$ we have $ a[b'] < w$. Then $(b,w)\in X$. If it were dominated by $(b,w')$ this would imply $a[b]=w'$. If it were dominated by $(b'>b,w')$ then $a[b']\ge w'\ge w$, a contradiction. So $(b,w)\in Dom(X)$. The computation clearly takes $O(k_1)$ time for allocating $a$ and backward scanning it in the second pass, plus $O(|X|)$ time for the first pass. Of course if ($k_1 > k_2)$ the same is performed on the black side.
\end{proof}

\section{Bicliques with respect to bimodular decomposition}
Let us now investigate how bicliques and bisizes behave with respect to the four recursive bimodular decomposition cases, and for the four nontrivial base cases (cycles and paths and their bipartite complements) of Theorem~\ref{thlozin}. For shortening the proofs, we consider here decomposition into only \emph{two} parallel, series or K+S parts (that may still be decomposable and thus not be components), while the canonical decomposition has arbitrary arity nodes, but they can be greedily split into binary nodes of the same type.

\subsection{Parallel case}

\begin{theorem}
	\label{parralel_thm}
	Let $G$ be a bipartite graph and $V=X\uplus Y$ where $X$ and $Y$ are nonempty and nonadjacent. Then $\D G = Dom( \D X \cup \D Y \cup \{(0,|W_G|),$ $(|B_G|,0)\} )$.
\end{theorem}

\begin{proof}
  For any nontrivial maxbisize $(b,w)$ of $G$, there exists a biclique $C$ of that size, included either in $X$ or in $Y$ (since they are nonadjacent). Let $M=X$ or $Y$ be the side $C$ is. Applying Proposition~\ref{mod_flag_incl}  we get that $(|B_M \cap C|,|W_M \cap C|) = (b,w)$ is a maxbisize of $M$, and thus, adding the two trivial maxbisizes, we have  $\D G \subseteq \D X \cup \D Y \cup \{(0,|W_G|),(|B_G|,0)\}$.
  Since a maxbisize of $X$ (resp. $Y$) is a bisize of $G$ by proposition~\ref{remincl} we have $\D G = Dom( \D X \cup \D Y \cup \{(0,|W_G|),(|B_G|,0)\} )$.
\end{proof}

\subsection{Series case}
\begin{theorem}
	\label{serie_thm}
	Let $G$ be a bipartite graph and $V=X\uplus Y$ where $X$ and $Y$ are nonempty and fully adjacent. Then $\D G = Dom(\D X \oplus \D Y)$
\end{theorem}

\begin{proof}
  For any nontrivial maxbisize $(b,w)$ of $G$, there exists a biclique $C$ of that size. Applying Property \ref{mod_flag} to $M=X$ we get that $(|B_X \cap C|,|W_X \cap C|)$ is a maxbisize of $X$, and to $M=Y$ that $(|B_Y \cap C|,|W_Y \cap C|)$ is a maxbisize of  $Y$. Since $(|B_X \cap C|,|W_X \cap C|) + (|B_Y \cap C|,|W_Y \cap C|) = (b,w)$, we have $\D G \subseteq \D X \oplus \D Y$.

  Taking the union of any $(b,w)$-sized biclique from $X$ and any $(b',w')$-sized biclique of $Y$   yields a  $(x+x', y+y')$-sized biclique of $G$, therefore $\D X \oplus \D Y$ contains only bisizes of $G$.

Applying Proposition\ref{remincl} yields  $\D G = Dom(\D X \oplus \D Y)$. 
\end{proof}

\subsection{K+S case}
\begin{theorem}
	\label{ks_thm}
	Let $G$ be a bipartite graph and $V=X\uplus Y$ where $X$ and $Y$ are nonempty and $X$ left adjacent to $Y$. Then $\D G = Dom( (\rightarrow_w^{|W_Y|} \D X) \cup (\rightarrow_b^{|B_X|} \D Y) )$.
\end{theorem}

\begin{proof}
	First let us prove  $\D G \subseteq (\rightarrow_w^{|W_Y|} \D X) \cup (\rightarrow_b^{|B_X|} \D Y)$. Let  $(b,w)$ be a maxbisize of $G$ and $C$ a biclique of that size. As there is no edge between $B_Y$ and $W_X$ then $C \cap W_X = \emptyset$ or $C \cap B_Y = \emptyset$.
	
	If we suppose $C \cap W_X = \emptyset$ then the size of $C\cap X$ is $(|B_X|,0)$ (otherwise  $B_X \cup (C \setminus X)$ would be a biclique strictly larger than $C$, impossible). 
	$C \setminus Y$ is fully adjacent to $Y$ so by 
	applying Property~\ref{mod_flag} with $M=Y$ we get that  $(|B_Y \cap C|,|W_Y \cap C|)$ is a maxbisize of $\D Y$.
	Then $(b,w) = (|B_Y \cap C|+|B_X|,|W_Y \cap C|)$ is in $\rightarrow_b^{|B_X|} \D Y$.
	
	If we suppose now $C \cap B_Y = \emptyset$ , then the size of $C \cap Y$ is $(0,|W_Y|)$ (otherwise,  $W_Y \cup (C \setminus Y)$ would be a biclique strictly larger than $C$, impossible). 
	$C \setminus X$ is fully adjacent to $X$ so by applying Property
	\ref{mod_flag} with $M=X$ we get that $(|B_X \cap C|,|W_X \cap C|)$ is a maxbisize of $\D X$. Then $(b,w) = (|B_X \cap C|,|W_X \cap C|+|W_Y|)$ is in $\rightarrow_w^{|W_Y|} \D X$. And finally $\D G \subseteq (\rightarrow_w^{|W_Y|} \D X) \cup (\rightarrow_b^{|B_X|} \D Y)$.
	
	Then notice that  $(b,w)\in \rightarrow_b^{|B_X|} \D Y$ is a bisize of $G$, since 
	the union of a biclique of size $(b-|B_x|,w)$ in $Y$ (which exists since  $(b-|B_x|,w)$ is a maxbisize of $Y$) and of $B_X$ is a biclique of $G$. The same is true for   $\rightarrow_w^{|W_Y|} \D X$ and finally applying Proposition\ref{remincl} $\D G = Dom( (\rightarrow_w^{|W_Y|} \D X) \cup (\rightarrow_b^{|B_X|} \D Y))$.
\end{proof}

\subsection{Prime case}
Now let us see how, in the Prime case, knowing the maxbisize set of the maximum canonical bimodules allows to compute the maxbisize set of a graph.

\begin{definition}[quotient graph and its maximal biclique set]
  Let $G$ be a twin-free graph having a Prime decomposition (i.e. no Parallel, Series nor K+S decomposition and at least 4 vertices), and $M_1...M_k$ is maximum canonical bimodules. $M_i$  may either consist in a non-trivial (i.e. of at least 4 vertices) canonical bimodule, or in a trivial one-vertex bimodule $\{v_i\}$ (when $v_i$ is either augmenting or does not belong to any non-trivial bimodule). The \emph{quotient graph}  $H_G=(V_{GB}\uplus V_{GW}, E_G)$ is defined as follow.
  There is a vertex $b_i\in V_{GB}$ iff $  M_i \cap B \neq \emptyset$.
  There is a vertex $w_i\in V_{GW}$ iff $  M_i \cap W \neq \emptyset$.
  There is an edge in $E_G$  between $b_i$ and $w_j$ iff there exists an edge between some vertex of $M_i\cap B$ and some vertex of $M_j \cap W$ (notice that, dealing with bimodules, using ``some'' or ``all'' in this definition is equivalent when $i \neq j$).

   $\mathcal{C_{H_G}}$ denotes the set of \emph{all} maximal bicliques of $H_G$. 
\end{definition}

If a graph has $k_1$ nontrivial maximal canonical bimodules and $k_2$ vertices not in any canonical bimodule then $H_G$ has $2k_1 + k_2$ vertices.
Let us now present two functions allowing to go between vertex-subsets of $G$ and of  $H_G$:

\begin{definition}[Corr(S)]
  For a vertex subset $S$ of $G$, let  $Corr(S)$ be a vertex subset of $H_G$ such that:

\begin{itemize}
	\item $b_i \in Corr(S)$ iff $S \cap M_i \cap B \neq \emptyset$
	\item $w_i \in Corr(S)$ iff $S \cap M_i \cap W \neq \emptyset$
\end{itemize}
\end{definition}

\begin{definition}[Rroc(S)]
  For  a vertex subset $S$ of $H_G$, let $Rroc(S)$ be the vertex subset of $G$ such that,  for any $x\in V$, $x\in Rroc(S)$ when
  
  \begin{itemize}
  \item there exists  $i\in [1,k]$ such that $x\in B_{M_i}$ and $b_i\in S$, or
  \item there exists  $i\in [1,k]$ such that $x\in W_{M_i}$ and $w_i\in S$
  \end{itemize}
\end{definition}


\begin{proposition}
	\label{bicliqueH}
If $C$ is a biclique of $G$ then $Corr(C)$ is a biclique of $H_G$. For each  biclique $C$ of $G$ there exist a maximal biclique $C'$ of $H_G$ such that $C \subseteq Rroc(C')$.
\end{proposition}

\begin{proof}
	Let $b_i,w_j \in Corr(C)$ for some $i,j \in [1,k]$. By definition there exist $b \in B_C \cap M_i$ and $w \in W_C \cap M_j$, since $C$ is a biclique $\{b,w\} \in E$ then there is an edge between $b_i$ and $w_j$.
	
	$Corr(C)$ is a biclique in $H_G$ therefore there exists a maximum biclique $C'$ of $H_G$ such that $Corr(C) \subseteq C'$. 
	let $b \in B_C$ and $w \in W_C$, so there exists $i,j \in [1,k]$ such that $b \in M_i$ and $w \in M_j$, so $b_i,w_j \in Corr(C)$, therefore $b_i,w_j \in C'$ and $b,w \in Rroc(C')$.
\end{proof}

\begin{definition}[Maxbisize set with respect to a maximal biclique]
Let $C$ be a maximal biclique of $H_G$ and $M$ a bimodule of $G$. Let $\DC{M}$ be:
\begin{itemize}
  \item if $C$ does not intersect $M$, $\emptyset$,
  \item otherwise, if $C\cap M \subset B$, then $(|M \cap B|,0)$,
  \item otherwise, if $C\cap M \subset W$, then $(0,|M \cap W|)$,
  \item otherwise, $\D{G[M]}$.
\end{itemize}
\end{definition}

\begin{theorem}
  \label{prem_thm}
  Let $k$ be a given constant integer. For any graph $G$ with a Prime decomposition, such that $G$ has at most $k$ maximal nontrivial canonical bimodules $M_1,...M_k$
  
  \begin{enumerate}
    \item for $C\in  \mathcal{C}_{H_G}$, $\D{Rroc(C)}=Dom(\DC{M_1}\oplus Dom(\DC{M_2} \oplus...\DC{M_k})...))$
    \item  $\D G = Dom( \bigcup_{C \in \mathcal{C}_{H_G}} \D {Rroc(C)})$ 
  \end{enumerate}
\end{theorem}

\begin{proof}
  Let  $C\in  \mathcal{C}_{H_G}$ be any maximal biclique of $H_G$. $G[Rroc(C)]$  falls in the Series case of Theorem~\ref{thlozin} or of Theorem~\ref{thfm}: each maximum canonical bimodule $M_i$ is either fully adjacent to the others or absent (not in $G[Rroc(C)]$, when $b_i\notin C$ and $w_i\notin C$). Then we just have to apply Theorem~\ref{serie_thm}. Just take care on what  $\D G[M_i]$ is :
  
\begin{itemize}
\item if  $b_i\notin C$ and $w_i\notin C$, then $M_i$ is absent from  $ G[Rroc(C)]$ and $\D G[M_i]$ is the empty set.
\item if $b_i\in C$ and ($w_i \notin C$ or $w_i$ does not exists) then  $M_i \cap B$ has only one maxbisize : $(|M \cap B|,0)$.
\item if $w_i \in C$ and ($b_i \notin C$ or $b_i$ does not exists) then $M_i \cap W$ has only one maxbisize : $(0,|M \cap W|)$.
\item otherwise (if both  $b_i \in C$ and $w_i\in C)$ then we assume we know $\D G[M_i]$. 
\end{itemize}

$\DC{M_i}$ is just defined so that the first assertion is true: $\D{Rroc(C)}=Dom(\DC{M_1}\oplus Dom(\DC{M_2} \oplus...\DC{M_k})...))$.

For the second assertion, let  $(b,w)$ be a maxbisize of $G$ and $D$ a biclique of that size.
By Proposition~\ref{bicliqueH} there must exists a maximal biclique $C$ in $H_G$ such that $D \subseteq Rroc(C)$.
Then $(b,w) \in \D{C}$ by Proposition~\ref{mod_flag_incl} and finally  $\D G \subseteq \bigcup_{C \in \mathcal{C}_{H_G}} \D {Rroc(C)}$.

An element of $\bigcup_{C \in \mathcal{C}_{H_G}} \D {Rroc(C)}$ is a maxbisize of a subgraph of $G$, so is a bisize of $G$, and finally by Proposition~\ref{remincl} $\D G = Dom( \bigcup_{C \in C_{H_G}} \D {Rroc(C)})$.
\end{proof}

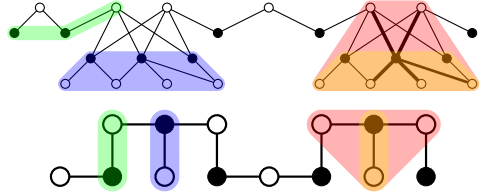
\begin{figure}
	\vspace{-.5cm}
	\begin{center}
		\resizebox{180pt}{!}{
			\begin{tikzpicture}[scale=1, transform shape]
				
				\newcommand{\convexpath}[2]{
					[   
					create hullnodes/.code={
						\global\edef\namelist{#1}
						\foreach [count=\counter] \nodename in \namelist {
							\global\edef\numberofnodes{\counter}
							\node at (\nodename) [draw=none,name=hullnode\counter] {};
						}
						\node at (hullnode\numberofnodes) [name=hullnode0,draw=none] {};
						\pgfmathtruncatemacro\lastnumber{\numberofnodes+1}
						\node at (hullnode1) [name=hullnode\lastnumber,draw=none] {};
					},
					create hullnodes
					]
					($(hullnode1)!#2!-90:(hullnode0)$)
					\foreach [
					evaluate=\currentnode as \previousnode using \currentnode-1,
					evaluate=\currentnode as \nextnode using \currentnode+1
					] \currentnode in {1,...,\numberofnodes} {
						let
						\p1 = ($(hullnode\currentnode)!#2!-90:(hullnode\previousnode)$),
						\p2 = ($(hullnode\currentnode)!#2!90:(hullnode\nextnode)$),
						\p3 = ($(\p1) - (hullnode\currentnode)$),
						\n1 = {atan2(\y3,\x3)},
						\p4 = ($(\p2) - (hullnode\currentnode)$),
						\n2 = {atan2(\y4,\x4)},
						\n{delta} = {-Mod(\n1-\n2,360)}
						in 
						{-- (\p1) arc[start angle=\n1, delta angle=\n{delta}, radius=#2] -- (\p2)}
					}
					-- cycle
				}
				
				\coordinate (P11) at (0,0);
				\coordinate (P10) at (2,0);
				\coordinate (P9) at (4,0);
				\coordinate (P8) at (6,0);
				
				\coordinate (J) at (1,1);
				\coordinate (I) at (3,1);
				\coordinate (H) at (5,1);
				
				\coordinate (P5) at (4,3);
				\coordinate (E) at (6,2);
				
				\coordinate (P6) at (2,3);
				\coordinate (F) at (0,2);
				\coordinate (P7) at (-1,3);
				\coordinate (G) at (-2,2);
				
				\draw[very thick] (G) -- (P7);
				\draw[very thick] (P7) -- (F);
				\draw[very thick] (P6) -- (F);
				
				\draw[very thick] (P6) -- (J);
				\draw[very thick] (P6) -- (I);
				\draw[very thick] (P6) -- (H);
				
				\draw[very thick] (P5) -- (J);
				\draw[very thick] (P5) -- (I);
				\draw[very thick] (P5) -- (H);
				
				\draw[very thick] (P5) -- (E);

				\draw[very thick] (P11) -- (J);
				
				\draw[very thick] (P10) -- (I);
				\draw[very thick] (P10) -- (J);
				
				\draw[very thick] (P9) -- (I);
				
				\draw[very thick] (P8) -- (I);
				\draw[very thick] (P8) -- (H);

				\filldraw [black] (J) circle (5pt);
				\filldraw [black] (I) circle (5pt);
				\filldraw [black] (H) circle (5pt);
				\filldraw [black] (E) circle (5pt);
				\filldraw [black] (F) circle (5pt);
				\filldraw [black] (G) circle (5pt);
				
				\filldraw [color=black, fill=white, very thick] (P5) circle (5pt);
				\filldraw [color=black, fill=white, very thick] (P6) circle (5pt);
				\filldraw [color=black, fill=white, very thick] (P7) circle (5pt);
				\filldraw [color=black, fill=white, very thick] (P8) circle (5pt);
				\filldraw [color=black, fill=white, very thick] (P9) circle (5pt);
				\filldraw [color=black, fill=white, very thick] (P10) circle (5pt);
				\filldraw [color=black, fill=white, very thick] (P11) circle (5pt);

				\fill[blue, opacity=0.3] \convexpath{P11,J,H,P8}{8pt};
				\fill[green, opacity=0.3] \convexpath{G,F,P6,F}{8pt};

\coordinate (aP11) at (10,0);
\coordinate (aP10) at (12,0);
\coordinate (aP9) at (14,0);
\coordinate (aP8) at (16,0);

\coordinate (aJ) at (11,1);
\coordinate (aI) at (13,1);
\coordinate (aH) at (15,1);

\coordinate (aP5) at (14,3);
\coordinate (aE) at (16,2);

\coordinate (aP6) at (12,3);
\coordinate (aF) at (10,2);
\coordinate (aP7) at (8,3);

\draw[very thick] (E) -- (aP7);
\draw[very thick] (aP7) -- (aF);
\draw[very thick] (aP6) -- (aF);

\draw[very thick] (aP6) -- (aJ);
\draw[line width=4pt] (aP6) -- (aI);
\draw[very thick] (aP6) -- (aH);

\draw[very thick] (aP5) -- (aJ);
\draw[line width=4pt] (aP5) -- (aI);
\draw[very thick] (aP5) -- (aH);

\draw[very thick] (aP5) -- (aE);

\draw[very thick] (aP11) -- (aJ);

\draw[line width=4pt] (aP10) -- (aI);
\draw[very thick] (aP10) -- (aJ);

\draw[line width=4pt] (aP9) -- (aI);

\draw[line width=3pt] (aP8) -- (aI);
\draw[very thick] (aP8) -- (aH);

\filldraw [black] (aJ) circle (5pt);
\filldraw [black] (aI) circle (5pt);
\filldraw [black] (aH) circle (5pt);
\filldraw [black] (aE) circle (5pt);
\filldraw [black] (aF) circle (5pt);

\filldraw [color=black, fill=white, very thick] (aP5) circle (5pt);
\filldraw [color=black, fill=white, very thick] (aP6) circle (5pt);
\filldraw [color=black, fill=white, very thick] (aP7) circle (5pt);
\filldraw [color=black, fill=white, very thick] (aP8) circle (5pt);
\filldraw [color=black, fill=white, very thick] (aP9) circle (5pt);
\filldraw [color=black, fill=white, very thick] (aP10) circle (5pt);
\filldraw [color=black, fill=white, very thick] (aP11) circle (5pt);

\fill[red,opacity=0.3] \convexpath{aP6,aP5,aP8,aP11}{8pt};
\fill[yellow, opacity=0.3] \convexpath{aP11,aJ,aH,aP8}{8pt};

\end{tikzpicture}}
		
	\end{center}\vspace{-.5cm}
	\begin{center}
	\resizebox{150pt}{!}{
		\begin{tikzpicture}[scale=1, transform shape]
			
			\newcommand{\convexpath}[2]{
				[   
				create hullnodes/.code={
					\global\edef\namelist{#1}
					\foreach [count=\counter] \nodename in \namelist {
						\global\edef\numberofnodes{\counter}
						\node at (\nodename) [draw=none,name=hullnode\counter] {};
					}
					\node at (hullnode\numberofnodes) [name=hullnode0,draw=none] {};
					\pgfmathtruncatemacro\lastnumber{\numberofnodes+1}
					\node at (hullnode1) [name=hullnode\lastnumber,draw=none] {};
				},
				create hullnodes
				]
				($(hullnode1)!#2!-90:(hullnode0)$)
				\foreach [
				evaluate=\currentnode as \previousnode using \currentnode-1,
				evaluate=\currentnode as \nextnode using \currentnode+1
				] \currentnode in {1,...,\numberofnodes} {
					let
					\p1 = ($(hullnode\currentnode)!#2!-90:(hullnode\previousnode)$),
					\p2 = ($(hullnode\currentnode)!#2!90:(hullnode\nextnode)$),
					\p3 = ($(\p1) - (hullnode\currentnode)$),
					\n1 = {atan2(\y3,\x3)},
					\p4 = ($(\p2) - (hullnode\currentnode)$),
					\n2 = {atan2(\y4,\x4)},
					\n{delta} = {-Mod(\n1-\n2,360)}
					in 
					{-- (\p1) arc[start angle=\n1, delta angle=\n{delta}, radius=#2] -- (\p2)}
				}
				-- cycle
			}
			
			\coordinate (P11P8) at (3,0);
			\coordinate (HIJ) at (3,1);
			
			\coordinate (P5) at (4,1);
			\coordinate (E) at (4,0);
			
			\coordinate (P6) at (2,1);
			\coordinate (F) at (2,0);
			\coordinate (P7) at (1,0);
			
			\draw[very thick] (P7) -- (F);
			\draw[very thick] (P6) -- (F);
			
			\draw[very thick] (P6) -- (HIJ);
			\draw[very thick] (P5) -- (HIJ);
			
			\draw[very thick] (P5) -- (E);

			\draw[very thick] (P11P8) -- (HIJ);
			
			\filldraw [black] (HIJ) circle (5pt);
			\filldraw [black] (E) circle (5pt);
			\filldraw [black] (F) circle (5pt);
			
			\filldraw [color=black, fill=white, very thick] (P5) circle (5pt);
			\filldraw [color=black, fill=white, very thick] (P6) circle (5pt);
			\filldraw [color=black, fill=white, very thick] (P7) circle (5pt);
			\filldraw [color=black, fill=white, very thick] (P11P8) circle (5pt);

			\fill[blue, opacity=0.3] \convexpath{P11P8,HIJ}{8pt};
			\fill[green, opacity=0.3] \convexpath{P6,F}{8pt};

				\coordinate (aP11P8) at (7,0);
				\coordinate (aHIJ) at (7,1);
				
				\coordinate (aP5) at (8,1);
				\coordinate (aE) at (8,0);
				
				\coordinate (aP6) at (6,1);
				\coordinate (aF) at (6,0);
				\coordinate (aP7) at (5,0);
				
				\draw[very thick] (aP7) -- (aF);
				\draw[very thick] (aP6) -- (aF);
				
				\draw[very thick] (aP6) -- (aHIJ);
				\draw[very thick] (aP5) -- (aHIJ);
				\draw[very thick] (aP7) -- (E);
				\draw[very thick] (aP5) -- (aE);

				\draw[very thick] (aP11P8) -- (aHIJ);
				
				\filldraw [black] (aHIJ) circle (5pt);
				\filldraw [black] (aE) circle (5pt);
				\filldraw [black] (aF) circle (5pt);
				
				\filldraw [color=black, fill=white, very thick] (aP5) circle (5pt);
				\filldraw [color=black, fill=white, very thick] (aP6) circle (5pt);
				\filldraw [color=black, fill=white, very thick] (aP7) circle (5pt);
				\filldraw [color=black, fill=white, very thick] (aP11P8) circle (5pt);

				\fill[red,opacity=0.3] \convexpath{aP6,aP5,aP11P8}{8pt};
				\fill[yellow, opacity=0.3] \convexpath{aP11P8,aHIJ}{8pt};
				
		\end{tikzpicture}}
	\end{center}\vspace{-.5cm}
	\caption{An example of $G$  with non-trivial maximal bimodules highlighted in green, yellow and blue; $H_G$ drawn below with the edges corresponding to non-trivial maximal bimodules colored the same color. A biclique of $G$ whose size is a maxbisize is drawn with thicker edges; the corresponding biclique $C$ in $H_G$ $C$ (bottom), and $Rroc(C)$ (top), are highlighted in red.}\vspace{-.5cm}
\end{figure}

\subsection{Base cases:  paths and cycles and their bipartite complements}
\begin{remark}\label{rempathcy}
  Let $G$ be cycle or a path of $b>2$ and $w>2$ vertices. $\D G=\{(1,2),(2,1),(0,w),(b,0)\}$ 
\end{remark}

\begin{theorem} \label{path2_thm}
Let $G$ be the bipartite complement of a path of $b>2$ and $w>2$ vertices. $\D G=\{(b,0), (0,w), (b',w')\ | \ (b'+w')= \lfloor \frac{n}{2} \rfloor\,\ b'\le b,\ w'\le w \}$ 
\end{theorem}

\begin{proof}
  Let us number from 1 to $n$ the vertices of $G$ along the path that is $\bc{G}$. Two vertices are adjacent iff their difference is odd and at least three. Given two consecutive vertices, at most one can belong to a given biclique. If we take all vertices from the same color we get the $(b,0)$ or $(0,w)$ bicliques. Otherwise, if a set contains both colors, a  set of  $\lfloor \frac{n}{2} \rfloor + 1$ or more vertices contains a pair of consecutive vertices and is not a biclique.
  Let us now show how to construct a $(b',w')$-sized biclique when 
$ (b'+w')= \lfloor \frac{n}{2} \rfloor$. If 1 is black, just take the $b'$ first black vertices of $G$, skip two vertices and take the remaining $w'$ white vertices. If $1$ is white, then take first $w'$ whites, skip two,  then $b'$ blacks. 
\end{proof}

\begin{theorem}
  \label{cycle_thm}
  Let $G$ be the bipartite complement of a cycle with $n$ vertices and $w \in W$ $b \in B$ such that $\{b,w\} \notin E$.
Then $\D{G} = Dom(\D{G - b} \cup \D{G - w})$.
\end{theorem}

\begin{proof}
	$\{b,w\} \notin E$ so any biclique is in $G - w$ or $G - b$, by Proposition~\ref{mod_flag_incl} $\D{G} \subseteq \D{G - b} \cup \D{G - w}$. 
	Any bisize of $G - w$ or $G - b$ is also a bisize of $G$ so by applying Proposition \ref{remincl} $\D{G} = Dom(\D{G - b} \cup \D{G - w})$.
\end{proof}

\section{Algorithms}
We shall see how to compute the maxbisize set for our two graphs classes. Then  Theorem~\ref{th3cl} says how the three biclique problems we address may be solved. But first let us state a complexity lemma.

\begin{lemma}
	\label{compl}
	Let $T$ be a tree with $n$ leave and no unary internal node. For a node $x$, let $|x|$ be the number of leaves in the subtree rooted at $x$. If, for each internal node $x$ with children  $x_1,x_2,...,x_k$ and for all $0 \leq i < j \leq k$ we perform an operation in  $O(|x_i| \times |x_j|)$ time, then the overall complexity is $O(n^2)$
\end{lemma}

\begin{proof}
  The number of leaves under a children $x_i$ is $|x_i|$. Performing an   $O(|x_i| \times |x_j|)$-time operation  for all $0 \leq i < j \leq k$ on each internal node $x$ with children  $x_1,x_2,...,x_k$ has the same complexity than performing an $O(1)$ time operation on each couple of leaves whose last common ancestor is $x$. Since there are exactly $n^2$ 
  couples of leaves, and each of them  may be affected to a unique last common ancestor, we get the announced complexity.
\end{proof}

\subsection{$Star_{1,2,3}$-free twin-free graphs }

\begin{theorem}
	Let $G$ be a twin-free  $Star_{1,2,3}$-free  graph.  $\D G$ can be computed in $O(n^2)$-time. Furthermore for a given maxbisize, a biclique of that size may be exhibited in $O(n^2)$ time.
\end{theorem}

\begin{proof}
We use Lozin Theorem (Theorem\ref{thlozin}, cf. \cite{lozin2002bipartite}). A decomposition tree may be computed in $O(n+m)$-time~\cite{quaddoura2002linear}. Then we adapt the tree so that is is binary by splitting arbitrarily each node (for instance a Series nodes with $k$ sons yields any subtree with $k-1$ internal binary Series nodes). And the maxbisize set is computed bottom-up along the tree. Let $N$ be a node, $M$ the set of leaves under it, and  $X$ and $Y$ its children if $N$ is not a leaf.
	
	\begin{itemize}
		\item if $N$ is a path or a cycle then, according to Remark~\ref{rempathcy}, $\D M = \{(2,1),(1,2),$ $(|B_M|,0),(0,|W_M|) \}$.
		\item  if $N$ is the bipartite complement of a path then, according to Theorem~\ref{path2_thm},  $\D M = \{(|B_M|,0), (0,|W_M|), (b',w')\ | \ (b'+w')= \lfloor \frac{|M|}{2} \rfloor\,\ b'\le |B_M|,\ w'\le |W_M| \}$
		
\item  if $N$ is the bipartite complement of a cycle with $w \in W$ $b \in B$ such that $wb \notin E$ then according to Theorem \ref{cycle_thm} $\D{N} = Dom(\D{N - b} \cup \D{N - w})$
		
		\item if $N$ is Parallel then, according to Theorem \ref{parralel_thm} $\D M = Dom( \D X \cup \D Y \cup \{(0,|W_M|),(|B_M|,0)\} )$.
		
	        \item if $N$ is Series then, according to Theorem \ref{serie_thm}   $\D M = Dom(\D X \oplus \D Y)$.
		
		\item if $N$ is K+S then, according to Theorem \ref{ks_thm} $\D M = Dom( (\rightarrow_w^{|W_Y|} \D X) \cup (\rightarrow_b^{|B_X|} \D Y) )$.
	\end{itemize}
        We notice that the operators used ($Dom$, $\cup$, $\oplus$, $\rightarrow_w$ and $\rightarrow_b$) may all be computed in $O(n^2)$ if maxbisize sets are kept as sorted lists. For Dom it is given by Lemma~\ref{lemDomn2}, for the other ones is is straight from definition. Then Lemma~\ref{compl} says the overall complexity is   $O(n^2)$.

        To exhibit a biclique of a given size, we need to retrieve (or to have memorized) for each maxbisize, the (at most two) maxbisizes used for adding it in the set. A backward computation yields  trivial maxbisizes as base cases, allowing to compute the biclique by taking any vertices from the nonzero color in the vertex set corresponding to each trivial maxbisize.
\end{proof}

\begin{corollary}
Let $G$ be a twin-free  $Star_{1,2,3}$-free  graph. A vertex-maximum biclique, and an edge-maximum biclique, and a maximum balanced biclique, may be computed in $O(n^2)$ time.
\end{corollary}

\subsection{Bounded Bimodularwidth twin-free graphs}

\begin{theorem}
	Let $k$ be a constant and $G$ be a twin-free bipartite graph of bimodularwith at most $k$ and  $T_G$ its canonical bimodular decomposition tree.  $\D G$ can be computed in $O(n^2)$-time. Furthermore for a given maxbisize, a biclique of that size may be exhibited in $O(n^2)$ time.
  \end{theorem}

\begin{proof}
We use the canonical decomposition theorem (Theorem\ref{thfm}, cf. \cite{WG04}). Like in the previous theorem, we adapt the tree so that each Series, Parallel or K+S node is split into a binary subtree, but we keep the Prime nodes with their at most $k$ children. Then  the maxbisize set is computed bottom-up along the tree. Let $N$ be a node, $M$ the set of leaves under it, and  $X$ and $Y$ its children if $N$ is Series, Parallel of K+S.
	
	\begin{itemize}
		\item if $N$ is a black (resp. a white) leaf then $\D M = \{(1,0)\}$  (resp. $ \{(0,1)\}$.
		\item if $N$ is Parallel then, according to Theorem \ref{parralel_thm},\\ $\D M = Dom( \D X \cup \D Y \cup \{(0,|W_M|),(|B_M|,0)\} )$
		
		\item 	if $N$ is Series then, according to Theorem \ref{serie_thm},\\ $\D M = Dom(\D X \oplus \D Y)$
		
		\item 	 $N$ is K+S then, according to Theorem \ref{ks_thm},\\ $\D M = Dom( (\rightarrow_w^{|W_Y|} \D X) \cup (\rightarrow_b^{|B_X|} \D Y) )$
		\item finally, if $N$ is Prime then, according to Theorem~\ref{prem_thm}, \\$\D M = Dom( \bigcup_{C \in \mathcal{C}_{H_M}} \D {Rroc(C)})$
	\end{itemize}
        The Series, Parallel and K+S case are proven like in the previous theorem.        Let us check the case $N$ is prime. Since $G$ has bounded bimodularwidth, it contains at most $k$ maximal canonical bimodules $M_1,...M_k'$ with $k'\le k$. The quotient $H_N$  has at most $2k$ nodes, and to list \emph{all} maximum bicliques of $G$ takes  $O(2^{2k})=O(1)$ time, and this list has $O(1)$ size with respect to $n$.  For a given maximal biclique $C$ of  $H_N$, Theorem~\ref{prem_thm} gives that  $\D{Rroc(C)}=Dom(\DC{M_1}\oplus Dom(\DC{M_2} \oplus...\DC{M_k'})...))$. That computation may be done in   $O(\sum_{x=0}^{i-1} |M_x| \times \sum_{y=x+1}^{i} |M_y|)$ time, and therefore Lemma~\ref{compl} applies: applied over all internal (Series, Parallel, K+S and Prime) nodes it yields the overall complexity is $O(n^2)$. 
	
        To exhibit a biclique of a given size, we may perform a backward computation like in the previous theorem. 
\end{proof}

\begin{corollary}
	Let $G$ be a bipartite graph of bimodularwith at most $k$ and  $T_G$ its canonical bimodular decomposition tree.  A vertex-maximum biclique, and an edge-maximum biclique, and a maximum balanced biclique, may be computed in $O(n^2)$ time.
\end{corollary}




\bibliographystyle{splncs04}
\bibliography{iwoca}

\end{document}